\documentclass[11pt]{article}

\usepackage[margin=2.5cm]{geometry}
\usepackage{helvet}
\usepackage{courier}
\usepackage[numbers,sort&compress]{natbib}
\usepackage{amsthm,amsmath,amssymb,bbm,graphicx,rotating}
\usepackage{boxedminipage}
\usepackage{enumerate}
\usepackage{epic,eepic}
\usepackage[rightcaption]{sidecap}
\usepackage{booktabs}
\usepackage{graphicx}
\usepackage{paralist}

\newcommand\blfootnote[1]{%
  \begingroup
  \renewcommand\thefootnote{}\footnote{#1}%
  \addtocounter{footnote}{-1}%
  \endgroup
}

\usepackage[textsize=scriptsize,textwidth=1.5cm]{todonotes}
\usepackage{hyperref}
\hypersetup{
  colorlinks=true,
  linkcolor=red!30!black,
  citecolor=green!30!black, 
  urlcolor=blue!30!black,
  bookmarksnumbered=true,
  pdflang={en},
  breaklinks,
  pagebackref
}

\newcommand{\doi}[1]{\url{http://dx.doi.org/#1}}

\usepackage[utf8]{inputenc}
\usepackage{microtype}
\usepackage[small]{caption}
\usepackage{booktabs}
\usepackage{algorithm}
\usepackage{algorithmic}
\usepackage{paralist}
\usepackage{caption}
\usepackage{subcaption}
\usepackage{wrapfig}
\usepackage{stmaryrd}

\newcommand{\SB}{\{\,}%
\newcommand{\SM}{\;{:}\;}%
\newcommand{\SE}{\,\}}%

\newcommand{\Card}[1]{|#1|}
\let\phi=\varphi
\let\epsilon=\varepsilon 
\def\hy{\hbox{-}\nobreak\hskip0pt} 
\newcommand{\parameter}[1]{\text{\normalfont{\sffamily #1}}}
 
\newcommand{\sol}{\parameter{sol}}

\newcommand{\sCSPD}{\text{\normalfont \#CSPD}}

\newcommand{\specialfont}[1]{{\normalfont\slshape #1}}

\newcommand{\ghtw}{\text{\specialfont{ghtw}}}
\newcommand{\ghtws}{\text{ghtw}}
\newcommand{\htw}{\text{\specialfont{htw}}}

\usepackage{verbatim}
\usepackage{enumitem}
\usepackage{xspace,url}

\usepackage{authblk}

\newtheorem{theorem}{Theorem}
\newtheorem{proposition}{Proposition}
\newtheorem{lemma}{Lemma}

\newcommand{\ctw}[1]{load-\ensuremath{#1} treewidth}
\newcommand{\ctd}[1]{load-\ensuremath{#1} tree decomposition}
\newcommand{\chtw}[1]{load-\ensuremath{#1} hypertree width}
\newcommand{\chtd}[1]{load-\ensuremath{#1} hypertree decomposition}
\newcommand{\colored}{loaded}
\newcommand{\black}{light}
\newcommand{\red}{heavy}

\newcommand{\cc}[1]{{\mbox{\textnormal{\textsf{#1}}}}\xspace}  %

\newcommand{\NP}{\cc{NP}}

\newcommand{\FPT}{\cc{FPT}}
\newcommand{\XP}{\cc{XP}}
\newcommand{\W}{{\cc{W}}}
\newcommand{\III}{{\mathcal{I}}}
\newcommand{\TTT}{{\mathcal{T}}}

\newcommand{\DDD}{{\mathcal{D}}}

\newcommand{\CSP}{\textsc{CSP}}

\newcommand{\Nat}{\mathbb{N}}

\newcommand{\bigoh}{\mathcal{O}}

\newcommand{\algtwxo}{\texttt{TW-X-Obl}}
\newcommand{\algtwxw}{\texttt{TW-X-W$\shortrightarrow$L}}
\newcommand{\algtwxl}{\texttt{TW-X-L$\shortrightarrow$W}}
\newcommand{\algtwho}{\texttt{TW-H-Obl}}
\newcommand{\algtwhw}{\texttt{TW-H-W$\shortrightarrow$L}}
\newcommand{\algtwhl}{\texttt{TW-H-L$\shortrightarrow$W}}
\newcommand{\alghtxo}{\texttt{HT-X-Obl}}
\newcommand{\alghtxw}{\texttt{HT-X-W$\shortrightarrow$L}}
\newcommand{\alghtxl}{\texttt{HT-X-L$\shortrightarrow$W}}
\newcommand{\alghtho}{\texttt{HT-H-Obl}}
\newcommand{\alghthw}{\texttt{HT-H-W$\shortrightarrow$L}}
\newcommand{\alghthl}{\texttt{HT-H-L$\shortrightarrow$W}}
\newcommand{\alghtgo}{\texttt{HT-G-Obl}}
\newcommand{\alghtgw}{\texttt{HT-G-W$\rightarrow$L}}

\author[1]{Andr\'e Schidler}
\author[1]{Robert Ganian}
\author[1,2]{Manuel Sorge}
\author[1]{Stefan Szeider}

\affil[1]{Algorithms and Complexity Group, TU Wien, Favoritenstrasse 9-11, 1040 Wien, Austria, \texttt{\{aschidler,rganian,manuel.sorge,sz\}@ac.tuwien.ac.at}}
\affil[2]{Faculty of Mathematics, Informatics and Mechanics, University of Warsaw, ul. Banacha 2, 02-097 Warsaw, Poland}

\title{Threshold Treewidth and Hypertree Width}

\begin{document}

\maketitle

\begin{abstract}
Treewidth and hypertree width have proven to be highly successful
structural parameters in the context of the Constraint Satisfaction
Problem (CSP). When either of these
parameters is bounded by a constant, then  CSP becomes solvable in polynomial time. However, here the order of the polynomial in the running time depends on the width, and this is known to be unavoidable; therefore, the problem is not fixed-parameter tractable parameterized by either of these width measures. Here we introduce an enhancement of tree and hypertree width through a novel notion of thresholds, allowing the associated decompositions to take into account information about the computational costs associated with solving the given CSP instance. Aside from introducing these notions, we obtain efficient theoretical as well as empirical algorithms for computing threshold treewidth and hypertree width and show that these parameters give rise to fixed-parameter algorithms for CSP as well as other, more general problems. We complement our theoretical results with experimental evaluations in terms of heuristics as well as exact methods based on SAT/SMT encodings.
\blfootnote{Preliminary and shortened versions of the results presented in this submission appeared in the proceedings of IJCAI 2020~\cite{GanianSSS20}. This article expands the exposition of that version by providing full proofs, detailed explanations especially including a more in-depth discussion of the applications of threshold treewidth, and an expanded experimental section. This article appeared in the \emph{Journal of Artificial Intelligence Research}~\cite{GanianSSS22}.}
\end{abstract}

\section{Introduction}

The utilization of structural properties of problem instances is a key approach to tractability of otherwise intractable problems such as Constraint Satisfaction, Sum-of-Products, and other hard problems that arise in AI applications~\cite{Dechter99,GottlobPichlerWei10,GottlobLeoneScarcello02}.
The idea is to represent the instance by a (hyper)graph and to exploit its decomposability to guide dynamic programming methods for solving the problem.
This way, one can give runtime guarantees in terms of the decomposition width.
The most successful width measures for graphs and hypergraphs are treewidth and hypertree width, respectively~\cite{GottlobGrecoScarcello14}.

\paragraph{Treewidth}
The Constraint Satisfaction Problem (CSP) can be solved
in time $d^k \cdot n^{\bigoh(1)}$ for instances whose primal graph has
$n$ vertices, treewidth $k$, and whose variables range over a domain of size~$d$~\cite{Dechter99,Freuder82}.  If~$d$ is a constant, then this
running time gives rise to fixed-parameter tractability w.r.t.\ the
parameter treewidth \cite{GottlobScarcelloSideri02}. However, without
such a constant bound on the domain size, it is known that CSP is
$\W[1]$\hy hard \cite{SamerSzeider10a} and hence not fixed-parameter
tractable.

In the first part of this paper, we propose a new framework that allows fixed-parameter tractability even if some variables range over large (though finite) domains.
The idea is to exploit tree decompositions with the special property that each decomposition bag contains only a few (say, at most $c$) such high-domain variables whose domain size exceeds a given threshold~$d$.
This results in a new parameter for CSP that we call the threshold-$d$ \ctw{c}.
We show that finding such tree decompositions is approximable to within a factor of $(c + 1)$ in fixed-parameter time, employing a replacement method which allows us to utilize state-of-the-art algorithms for computing treewidth such as Bodlaender \emph{et al.}'s approximation~\cite{BodlaenderDDFLP16}.
We then show that for any fixed $c$ and~$d$, CSP parameterized by threshold-$d$ \ctw{c} is fixed-parameter tractable, and that the same tractability result can be lifted to other highly versatile problems such as CSP with Default Values~\cite{GanianKimSlivovskySzeider18,GanianKimSlivovskySzeider21}, Valued CSP~\cite{SchiexFV95,Zivny12}, and the Integer Programming (IP) problem~\cite{Schrijver99}.

\paragraph{Hypertree width}
Bounding the treewidth of a CSP instance automatically bounds the
arity of its constraints. More general structural restrictions that
admit large-arity constraints can be formulated in terms of the
hypertree width of the constraint hypergraph. It is known that for any
constant~$k$, hypertree decompositions of width at most $k$ can be
found in polynomial time, and that CSP instances of hypertree width
$k$ can be solved in polynomial time. If $k$ is a parameter and not
constant, then both problems become $\W[1]$\hy hard and hence not
fixed-parameter tractable. We show that also in the context of
hypertree width, a more fine-grained parameter, which we call
\emph{threshold-$d$ \chtw{c}}, can be used to achieve
fixed-parameter tractability.  Here we distinguish between \red\ and \black\ hyperedges, where a hyperedge is light if the corresponding
constraint is defined by a constraint relation that contains at most
$d$ tuples. Each bag of a threshold-$d$ \chtd{c}{}
of width $k$ must admit an edge cover that consists of at most $k$
hyperedges, where at most $c$ of them are \red. We show that for any
fixed $c$ and~$k$, we can determine for a given hypergraph in
polynomial time whether it admits a hypertree decomposition of width
$k$ where the cover for each bag consists of at most $c$ \red\
hyperedges\footnote{This is not fixed-parameter tractable for
  parameter $k$, as already without the $c$ restriction, the problem
  is $\W[2]$\hy hard.}. We further show that for any fixed $c$ and
$d$, given a width-$k$ threshold-$d$ \chtd{c}\ of
a CSP instance, checking its satisfiability is fixed-parameter
tractable when parameterized by the width~$k$.

\paragraph{Practical algorithms and experiments}
The most popular practical algorithms for finding treewidth and
hypertree decompositions are based on characterizations in terms of
\emph{elimination orderings}. We show how these characterizations can
be extended to capture threshold treewidth and threshold hypertree
width.  These then allow us to obtain practical
algorithms that we test on large sets of graphs and hypergraphs
originating from real-world applications. In particular, we propose
and test several variants of the well-known min-degree heuristics, as
well as exact methods based on SMT-encodings for computing threshold 
tree and hypertree decompositions.  Our experimental findings are
significant, as they show that by optimizing decompositions towards
low load values we can obtain in many cases decompositions that are
expected to perform much better in the dynamic programming phase than
ordinary decompositions that are oblivious to the weight of vertices
or hyperedges.

\paragraph{Related work}
There are several reports on approaches for tuning
greedy treewidth heuristics to improve the performance of particular
dynamic programming (DP) algorithms.
For instance, \citet{JegouT17} considered computing tree decompositions whose bags induce connected subgraphs in order to speed up solution methods whose running time depends on the connected components induced by bags. 
\citet{KaskGelfandDechter11} optimized the state space of graphical
models for probabilistic reasoning, which corresponds in our setting
to minimizing the product of the domain sizes of variables that appear
together in a bag.
Similar heuristics were suggested by \citet{BachooreBodlaender07} for treewidth.
\citet{AbseherMusliuWoltran17} optimized heuristic tree decompositions
w.r.t.\ the sizes of DP tables when solving individual combinatorial
problems such as 3-Colorability or Minimum Dominating Set.
\citet{ScarcelloGrecoLeone07} presented a general framework for
minimizing the weight of hypertree decompositions of bounded width.
We discuss in Sections~\ref{section:tw} and~\ref{section:htw} how the above notions give rise to complexity
parameters for CSP and how they compare to threshold treewidth and
hypertree width.

\paragraph{Outline}
We give the basic definitions and notation in Section~\ref{sec:prelim}.
In Section~\ref{section:tw} we formally introduce the notion of threshold-$d$ \ctw{c} and give results on computing the associated decompositions.
In Section~\ref{sec:tw-appl} we give applications of these new notions to further prominent problems different from CSP in the AI context.
In Section~\ref{section:htw} we then introduce threshold-$d$ \chtw{c} and give results on computing the associated decompositions.
In Section~\ref{sec:elim-order} we give alternative characterizations of the threshold treewidth and hypertree width notions via so-called elimination orderings which we use in our experiments.
The algorithms we implemented are described in Section~\ref{sec:practical} and in Section~\ref{sec:exp} we report on the empirical results.
Section~\ref{sec:concl} contains a conclusion.

\section{Preliminaries}\label{sec:prelim}

For an integer $i$, we let $[i]=\{1,2,\dots,i\}$ and $[i]_0=[i]\cup\{0\}$.
We let $\Nat$ be the set of natural numbers, and $\Nat_0$ the set
$\Nat \cup \{0\}$. 
We refer to %
\citet{Diestel12} for
standard graph terminology.

Similarly to graphs, a \emph{hypergraph} $H$ is a pair $(V,E)$ where $V$ or
$V(H)$ is its vertex set and $E$ or $E(H)\subseteq 2^V$ is its set of hyperedges.
An \emph{edge cover} of $S\subseteq V$ (in the hypergraph $(V,E)$) is a set $F\subseteq E$ such that for every $v\in S$ there is some $e\in F$ with $v\in e$. The \emph{size} of an
edge cover is its cardinality. For a (hyper)graph $G$, we will sometimes use $V(G)$ to denote its vertex set and $E(G)$ to denote the set of its (hyper)edges.

\paragraph{Parameterized complexity}
In parameterized
algorithmics~\cite{DowneyFellows13,Niedermeier06,CyganFKLMPPS15,FlumGrohe06}, the
running-time of an algorithm is studied with respect to a parameter
$k\in\Nat_0$ and input size~$n$. The basic idea is to find a parameter
that describes the structure of the instance such that the
combinatorial explosion can be confined to this parameter. In this
respect, the most favorable complexity class is \FPT
(\textit{fixed-parameter tractable}), which contains all problems that
can be decided by an algorithm running in time $f(k)\cdot
n^{\bigoh(1)}$, where $f$ is a computable function. Algorithms with
this running-time are called \emph{fixed-parameter algorithms}. 
A less
favorable outcome is an \XP{} \emph{algorithm}, which is an algorithm
running in time $\bigoh(n^{f(k)})$; problems admitting such
algorithms belong to the class \XP.
Problems hard for the complexity classes $\W[1]$, $\W[2]$, \dots, $\W[P]$ do not admit fixed-parameter algorithms (even though they might be in \XP) under standard complexity assumptions.

\begin{figure}[tb]
  \centering
  \includegraphics[width=\textwidth]{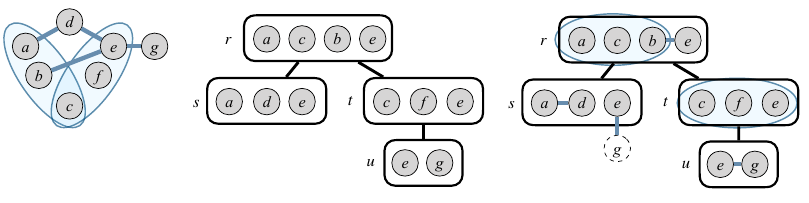}
  \caption{Left: a hypergraph $H$. Middle: a tree decomposition of $H$
    of width 3. Right: a hypertree decomposition of $H$ of width 2; the covers of the bags are indicated by the blue edges and blue encircled vertex sets.
    Observe that the hypertree decomposition satisfies the Special
    Condition: the bag at node $s$ is the only bag whose edge cover
    uses an edge containing a vertex from outside the bag (the edge
    $\{e,g\}\in \lambda(s)$ contains the vertex $g$ outside $\chi(s)$).
    However, as $s$ has no descendants, the Special Condition is
    trivially satisfied.}
  \label{fig:htw}
\end{figure}

\paragraph{Treewidth}
A \emph{tree decomposition}~$\mathcal{T}$ of a (hyper)graph $G$ is a pair 
$(T,\chi)$, where $T$ is a tree and $\chi$ is a function that assigns each 
tree node $t$ a set $\chi(t) \subseteq V(G)$ of vertices such that the following 
conditions hold:
\begin{compactenum}[(P1)]
\item[(P1)] For every (hyper)edge $e\in E(G)$ there is a tree node
  $t$ such that $e\subseteq \chi(t)$.
\item[(P2)] For every vertex $v \in V(G)$,
  the set of tree nodes $t$ with $v\in \chi(t)$ induces a non-empty subtree of~$T$.
\end{compactenum}
The sets $\chi(t)$ are called \emph{bags} of the decomposition~$\mathcal{T}$, and $\chi(t)$ 
is the bag associated with the tree node~$t$. 
The \emph{width} of a tree decomposition $(T,\chi)$ is the size of a largest bag minus~$1$.
The \emph{treewidth} of a (hyper)graph $G$,
denoted by $\textup{tw}(G)$, is the minimum
width over all tree decompositions of~$G$. 
\paragraph{Hypertree width}
\looseness=-1
A \emph{generalized hypertree decomposition} of a hypergraph $H$ is a triple
$\DDD=(T,\chi,\lambda)$ where $(T,\chi)$ is a tree decomposition of
$H$ and $\lambda$ is function mapping each $t\in V(T)$ to an
edge cover $\lambda(t)\subseteq E(H)$ of $\chi(t)$. The
\emph{width} of $\DDD$ is the size of a largest edge cover
$\lambda(t)$ over all $t\in V(T)$, and the generalized hypertree width
$\ghtw(H)$ of $H$ is the smallest width over all generalized hypertree
decompositions of~$H$.  

It is known to be \NP-hard to decide whether a
given hypergraph has generalized hypertree width $\leq 2$
\cite{FischlGottlobPichler18}. To make the recognition of
hypergraphs of bounded width tractable, one needs to strengthen the
definition of generalized hypertree width by adding a further
restriction. A \emph{hypertree decomposition}~\cite{GottlobLeoneScarcello02} of
$H$ is a generalized hypertree decomposition $\DDD=(T,\chi,\lambda)$
of $H$ where $T$ is a rooted tree that satisfies in addition to (P1) and (P2)
also the following  Special Condition~(P3):
\begin{compactenum}[(P3)]
\item[(P3)] If $t,t'\in V(T)$ are nodes in $T$ such that $t'$ is
  a descendant\footnote{A \emph{descendant} of a node $t$ in a tree $T$ is any node $t'$ on a path from $t$ to a leaf of~$T$ in the subtree rooted at~$t$.} of $t$, then for each $e\in \lambda(t)$ we have
$(e\setminus \chi(t))\cap \chi(t')=\emptyset$.
\end{compactenum}
The \emph{hypertree width} $\htw(H)$ of $H$ is the smallest width over
all hypertree decompositions of~$H$.

To avoid trivial cases, we consider only hypergraphs $H=(V,E)$ where
each $v\in V$ is contained in at least one $e\in E$. Consequently,
every considered hypergraph $H$ has an edge cover, and the parameters
$\ghtw(H)$ and $\htw(H)$ are always defined. If $\Card{V}=1$ then
$\htw(H)=\ghtw(H)=1$.

Figure~\ref{fig:htw} shows a hypergraph, a tree decomposition, and a
hypertree decomposition.

\paragraph{The constraint satisfaction problem}
\looseness=-1
An instance of a \emph{constraint satisfaction problem} (CSP) $\III$
is a triple $(V,D,C)$ consisting of a finite set $V$ of variables, a function $D$ which maps each variable $v\in V$ to a set (called the \emph{domain} of $v$), and
a set $C$ of constraints.
A \emph{constraint} 
$c \in C$ consists of a \emph{scope},
denoted by $S(c)$, which is a completely ordered subset of
$V$, and a relation, denoted by $R(c)$, which is a
$|S(c)|$-ary relation on $\Nat$.
If not stated otherwise, we assume that for each scope there is at most one constraint with that scope.
The \emph{size} of an instance $\III$ is $|\III|=|V|+|D|+\sum_{c\in C}|S(c)|\cdot|R(c)|$.

An \emph{assignment} is a mapping $\theta : V \rightarrow \Nat$ which maps each variable $v\in V$ to an element of $D(v)$; a partial assignment is defined analogously, but for $V'\subseteq V$. A constraint $c\in C$ with scope $S(c)=(v_1,\dotsc,v_{|S(c)|})$ is satisfied by a partial assignment $\theta$ if $R(c)$ contains the tuple $\theta(S(c))=(\theta(v_1),\dotsc,\theta(v_{|S(c)|}))$. An assignment is a \emph{solution} if it satisfies all constraints in $\III$. The task in \textsc{CSP} is to decide whether the instance $\III$ has at least one solution.

The \emph{primal graph} $G_\III$ of a \textsc{CSP} instance $\III=(V,D,C)$ is the graph whose vertex set is $V$ and where two vertices~$v,w$ are adjacent if and only if there exists a constraint whose scope contains both $v$ and $w$. The \emph{hypergraph $H_\III$} of $\III$ is the hypergraph with vertex set~$V$, where there is a hyperedge $E\subseteq V$ if and only if there exists a constraint with scope~$E$.
Note that the hypergraph does not contain parallel edges as for each scope there is at most one constraint with that scope.

\section{Threshold Treewidth}\label{section:tw}

The aim of this section is to define threshold treewidth for
\textsc{CSP}, but to do that we first need to introduce a refinement
of treewidth on graphs. Let $G$ be a graph where $V$ is bipartitioned
into a set of \black\ vertices and a set of \red\
vertices; we call such graphs \emph{\colored}. For
$c\in \Nat_0$, a \emph{\ctd{c}} of $G$ is a tree decomposition
of $G$ such that each bag $\chi(t)$ contains at most $c$ \red\ 
vertices. It is worth noting that, while every graph admits a
tree decomposition, for each fixed $c$ there are \colored\ graphs which do
not admit any \ctd{c}\ (consider, e.g., a complete graph
on $c+1$ \red\ vertices). The \emph{\ctw{c}} of $G$ is the minimum width
of a \ctd{c}\ of $G$ or $\infty$ if no such decomposition exists.

Let $d,c\in \Nat_0$ and $\III=(V,D,C)$ be a \textsc{CSP}
instance. Moreover, let $G^d_\III$ be the primal graph such that $v\in
V$ is \black\ if and only if $|D(v)|\leq d$. Then the
\emph{threshold-$d$ \ctw{c}} of $\III$ is the \ctw{c} of $G^d_\III$. 
The following theorem summarizes the  key advantage of using
the threshold-$d$ \ctw{c} instead of the ``standard'' treewidth of
$G_\III$.

\begin{theorem}
\label{thm:usingctw}
Given $d,c\in \Nat$, a \textsc{CSP} instance $\III$ and a \ctd{c}\ of $G^d_\III$ of width $k$, it is possible to solve $\III$ in time at most $\bigoh(d^{k + 1}\cdot |\III|^{c+2})$.
\end{theorem}

\newcommand{\cons}{q}

\begin{proof}
  The proof follows by applying the classical algorithm for solving
  \textsc{CSP} by using the treewidth of the primal graph
  $G_\III$~\cite{Freuder82,GottlobScarcelloSideri02}, whereas the
  stated runtime follows from the bound on high-domain variables
  imposed by the definition of \ctw{c}. However, since the proof idea
  is also used in the subsequent Propositions~\ref{prop:application-weighted-csp} to~\ref{prop:application-IP}, we provide a
  full description of the algorithm below for completeness.

Let $\TTT=(T,\chi)$ be the \ctd{c}\ of $G^d_\III$ provided on the input.
Choose an arbitrary node $t$ of $T$ and denote it as the root $r$. Let $V_t=\SB v\in V \SM v\in\chi(t) \vee($ there is a child $t'$ of $t$ such that $v\in \chi(t')~)\SE$. Moreover, let a \emph{$t$-mapping} be a mapping that assigns to each variable $v$ in $\chi(t)$ a value from $D(v)$. It is easy to see that the number of $t$-mappings is upper-bounded by $d^{k-c + 1}\cdot |\III|^{c}$.

The algorithm proceeds by computing, for each node $t$ in a leaf-to-root fashion, the set $M(t)$ of all $t$-mappings with the following property: $\theta\in M(t)$ if and only if there exists an extension $\theta'$ of $\theta$ to $V_t$ such that each constraint $\cons$ with $S(\cons)\subseteq V_t$ is satisfied by $\theta'$. Clearly, $\III$ is a YES-instance if and only if $M(r)$ is non-empty; moreover, if we correctly compute a non-empty $M(r)$ by leaf-to-root dynamic programming, then it is possible to reconstruct a solution for $\III$ by retracing the steps of the dynamic program in a standard fashion.

To compute $M(\ell)$ for a leaf $\ell$, it suffices to loop over all $\ell$-mappings and for each perform a brute-force check to determine whether all of the relevant constraints are satisfied.
For a non-leaf node $t$, we also loop over all $t$-mappings, whereas for each $t$-mapping $\theta$ we first check whether each constraint $c$ such that $S(c)\subseteq \chi(t)$ is satisfied; if not, we discard $\theta$. If yes, we then check that $\theta$ is ``consistent'' with each of the children of $t$---notably, for each child $t'$ of $t$, we ensure that there is at least one $t'$-mapping $\theta'$ such that $\forall v\in \chi(t)\cap \chi(t'): \theta'(v)=\theta(v)$.\footnote{This check can be carried out in amortized constant time via suitable data structures if all $t$-mappings are ordered based on a fixed variable ordering.} If this is the case then we add $\theta$ to $M(t)$. 

Correctness follows by the observation that each constraint $c$ such that $S(c)\subseteq V_t$ must be contained in a bag of at least one descendant of $t$, and hence each such constraint is checked against $\theta$ by transitivity. The runtime bound follows by the upper bound on $V(T)$ and the upper bound on the number of $t$-mappings for each node $t$. 
\end{proof}

We now briefly discuss the relation between threshold-$d$ \ctw{c}\ and other parameters of CSP instances related to treewidth and domain size.
First, \citet{BachooreBodlaender07} introduced a parameter called weighted treewidth.
Consider a graph~$G$ with vertex-weight function~$w \colon V(G) \to \Nat$.
The \emph{weighted width} of a tree decomposition $(T, \chi)$ of $G$ is $\max_{t \in V(T)}\Pi_{v \in \chi(t)}w(v)$, and the minimum such quantity is the \emph{weighted treewidth} of~$G$.
The \emph{weighted treewidth} of a CSP instance is the weighted treewidth of its primal graph with weight function~$w$ defined as $w(v) = |D(v)|$ for each variable~$v$.
It is not hard to see that we can replace the given \ctd{c}\ in Theorem~\ref{thm:usingctw} by a tree decomposition minimizing the weighted treewidth, say the minimum is~$w$, and the algorithm would run in $\bigoh(w \cdot |\III|^2)$ time.
However, %
the weighted treewidth implicitly upper-bounds the domains of \emph{all} variables.
This is not the case for \ctw{c}, which allows each bag to contain up to $c$ variables of arbitrarily large domains. Thus, \ctw{c}\ can be thought of as a more general parameter, that is, fixed-parameter algorithms for it apply to a larger set of instances.

Another way of dealing with variables with large domain would be to replace each of these variables~$v$ in every constraint by $\lceil\log |D(v)| \rceil$ \emph{representative} variables with domain size two.
Since the representative variables occur together in a constraint, they induce a clique in the primal graph.
Computing a tree decomposition of low width for the primal graph thus roughly corresponds to minimizing the number of high-domain variables in a bag.
More precisely, it corresponds to minimizing the sum of the logarithms of the domain sizes of the high-domain variables in the bags.
Similarly to weighted treewidth, this means that the (maximum) domain size is in a strong relation with the width.
In comparison, the approach taken here is aimed at restricting the number of high-domain variables that occur together in a bag.

To apply Theorem~\ref{thm:usingctw} it is necessary to be
able to compute a \ctd{c} of a \colored\ graph efficiently. While there is a significant body of literature on computing or approximating optimal-width tree decompositions of a given graph, it is not obvious how to directly enforce a bound on the number of \red\ vertices per bag in any of the known state-of-the-art algorithms for the problem.
Our next aim is to show that in spite of this, it is possible to reduce the problem of computing an approximate \ctd{c}\ to the problem of computing an optimal-width tree decomposition of a graph. This then allows us to use known results in order to find a sufficiently good approximation of \ctw{c}. 

\begin{lemma}%
\label{lem:approxctw}
\looseness=-1
Given an $n$-vertex \colored\ graph $G$ with $m$ edges and an integer $k \geq 1$, it is possible to compute in $\bigoh((n+m)\cdot k^2)$ time a graph $G'$ such that: %
\begin{inparaenum}  
\item[\emph{(1)}] If $G$ has \ctw{c} $k$ then $G'$ has treewidth at most $ck+k$, and 
\item[\emph{(2)}] given a tree decomposition of width $\ell$ of $G'$, in linear time we can compute a \ctd{(\ell/(k + 1))} of $G$ of width~$\ell$.
\end{inparaenum}
\end{lemma}

\begin{proof}
Consider the graph $G'$ constructed as follows: 
\begin{inparaenum}[(a)]
\item we add each \black\ vertex in $G$ into~$G'$;
\item for each \red\ vertex $v\in V(G)$, we add $k+1$ vertices $v_0, v_1,\dots, v_k$ into $G'$ (we call them \emph{images} of $v$);
\item we add an edge between each pair of images, say $v_i,v_j\in V(G')$, of some vertex $v$;
\item for each $vw\in E(G)$, we add into $G'$ the edge $vw$ (if both $v$ and $w$ are \black), or the edges $\SB vw_i \SM i\in [k]_0\SE$ (if $w$ was \red\ and $v$ was \black), or the edges $\SB v_iw_j \SM i,j\in [k]_0\SE$ (if both $v$ and $w$ were \red).
\end{inparaenum}

Clearly, $G'$ can be constructed from $G$ in time $\bigoh((n + m)\cdot k^2)$. For the part~(1) of the lemma, consider a minimum-width tree decomposition $\mathcal{T}=(T,\chi)$ of $G$. Now consider the mapping $\chi'$ that is obtained from $\chi$ by replacing each occurrence of a \red\ vertex $v$ by all of its images, i.e., $v_0,\dots,v_k$---formally, $x\in \chi'(t\in V(T))$ if and only if either $x\in \chi(t)$, or there exists $v\in V(G)$ such that $v\in \chi(t)$ and $x=v_i$. Since the number of \red\ vertices in a single bag was upper-bounded by~$c$, the maximum size of an image of $\chi'$ is $(k+1)\cdot c+k+1-c=ck+k+1$. It is easy to verify that $(T,\chi')$ is a tree decomposition of $G'$, and so the first claim follows.

\looseness=-1 For part~(2) of the lemma, call a tree decomposition $\mathcal{T'}=(T,\chi')$ of $G'$ \emph{discrete} if for each $v \in V(G)$ such that $v$ is heavy and each $t \in V(T)$ it holds that either for all $i \in [k]_0$ we have $v_i \notin \chi'(t)$ or for all $i \in [k]_0$ we have $v_i \in \chi'(t)$.
Let $\widehat{\mathcal{T}}=(T,\widehat{\chi})$ be a tree decomposition of $G'$ of width at most~$\ell$.
We first claim that in linear time we can compute a tree decomposition $\mathcal{T'}=(T,\chi')$ of $G'$ that is discrete and of width at most~$\ell$.
To do this, we compute $\chi'$ from $\widehat{\chi}$ as follows.
We iterate over all $t \in V(T)$ and for each vertex in $\widehat{\chi}(t)$ we check whether it is the image of some \red\ vertex $v \in V(G)$ and, if so, we check whether all images of $v$ are contained in $\widehat{\chi}(t)$.
If not all images of $v$ are contained in $\widehat{\chi}(t)$ we remove from $\widehat{\chi}(t)$ all images of $v$.
In this way we obtain a mapping $\chi'$.
Note that, for each $t$, the above computation can be done in $\bigoh(|\widehat{\chi}(t)|)$ time as follows.
First, iterate over $\widehat{\chi}(t)$, obtaining a list of \red\ vertices which have images in $\widehat{\chi}(t)$.
For each such vertex~$v$, initialize an empty list of images in $\widehat{\chi}(t)$.
Iterate over $\widehat{\chi}(t)$ again to fill the lists of images with pointers to the images in $\widehat{\chi}(t)$.
Finally, compute the length of each list and, if it is shorter than $k + 1$, remove all images from $\widehat{\chi}(t)$ using the pointers.
Thus, $(T, \chi')$ can be computed in linear time.

Next, we argue that $(T, \chi')$ is a tree decomposition of~$G'$.
Consider first condition~(P2) of tree decompositions.
Clearly, (P2) holds for every vertex $v$ which is not an image of a \red\ vertex.
For the sake of contradiction, assume that (P2) is violated for an image $v_j$, $j \in [k]_0$, of some \red\ vertex~$v \in V(G)$.
Thus, there are $r, s, t \in V(T)$ such that $s$ is on the unique path between $r$ and $t$ in $T$, $v_j \in \chi'(r)$, $v_j \in \chi'(t)$, and $v_j \notin \chi'(s)$.
Observe that both $\widehat{\chi}(r)$ and $\widehat{\chi}(t)$ contain all images of $v$ whereas there is an image $v_i$ of $v$ which is not contained in~$\widehat{\chi}(s)$.
Hence, (P2) is violated for $(T, \widehat{\chi})$ and vertex $v_i$, a contradiction.

Now consider condition~(P1).
Clearly, (P1) holds for each edge whose endpoints either both are images of a \red\ vertex of $G$ or both are not images of \red\ vertex of~$G$.
For the sake of contradiction, assume that (P1) does not hold for an edge such that one endpoint, $u$, is not the image of a \red\ vertex and one endpoint, $v_j$ for some $j \in [k]_0$, is the image of a \red\ vertex~$v \in V(G)$.
Since the images of $v$ induce a clique in $G$, there is a node $t \in V(T)$ such that $\widehat{\chi}(t)$ contains all images of~$v$.\footnote{This is a well-known fact about cliques and tree decompositions and can be proved roughly as follows: The vertices in the clique induce subtrees of the decomposition tree whose vertex sets have pairwise nonempty intersection. Since the trees are subtrees of  the decomposition tree, this means there is a vertex in the decomposition tree that is contained in all of the subtrees.}  
By assumption on $u$, we have $u \notin \chi'(t)$ and thus $u \notin \widehat{\chi}(t)$.
There is thus an edge $e$ in $T$ whose removal separates $T$ into a connected component that contains~$t$ and a connected component that contains all $t' \in T$ such that $u \in \chi'(t')$.
Moreover, there is such an edge $e$ such that one endpoint, $s$, has the property that $u \in \chi'(s)$.
Since $u$ is adjacent to each $v_i \in V(G')$, $i \in [k]_0$, for each $i \in [k]_0$ there is $r_i \in V(T)$ such that both $u, v_i \in \widehat{\chi}(r_i)$.
By (P2) of $(T, \widehat{\chi})$, for each $i \in [k]_0$, the subtree of $T$ induced by the nodes~$r \in V(T)$ with $v_i \in \widehat{\chi}(r)$ contains~$e$.
Thus, $\widehat{\chi}(s)$ contains each $v_i$.
By construction of $\chi'$ it follows that $\chi'(s)$ contains each $v_i$.
This is a contradiction to the fact that $\chi'(s)$ contains $u$ and to the assumption that there is no bag of $(T, \chi')$ that contains both $u$ and $v_j$.
Thus, (P2) holds for $(T, \chi')$.

Above we have shown that the discrete tree decomposition $(T, \chi')$ of $G'$ of width $\ell$ can be computed in linear time.
Next, let us compute the mapping $\chi$ from $\chi'$ as follows:
For each $t \in V(T)$, we put $v\in \chi(t)$ if either $v\in \chi'(t)$ or there exists $j\in [k]_0$ such that $v_j\in \chi'(t)$. Since in this way each vertex in a bag $\chi'(t)$ can only lead to the addition of at most one vertex into $\chi(t)$, it is easy to see that the maximum size of an image of $\chi$ is $\ell+1$.
Hence, if $(T, \chi)$ is a tree decomposition, then its width is at most~$\ell$.

We claim that the load of $(T, \chi)$ is at most $\ell/(k + 1)$.
Otherwise, there would be some $t \in V(T)$ such that $\chi(t)$ contains $\ell/(k + 1) + 1$ \red\ vertices.
In that case, by discreteness of $\chi'$, the number of vertices in $\chi'(t)$ is at least $(k + 1) \cdot (\ell/(k + 1) + 1) = \ell + k + 1 > \ell + 1$.
This contradicts the fact that $(T, \chi')$ has width $\ell$.

It remains to show that $(T,\chi)$ is a tree decomposition of $G$. 
Condition (P1) clearly holds for every edge $vw\in E(G)$ such that $vw\in E(G')$. On the other hand, if $vw\not \in E(G')$ then either one or both of $v,w$ are \red\ in $G$, and hence, e.g., the vertices $v_0$ and $w_0$ are adjacent in $G'$. This implies that there is some node $t'\in V(T)$ such that $\{v_0,w_0\}\subseteq \chi'(t')$, and by construction we obtain $\{v,w\}\subseteq \chi(t')$---hence (P1) holds. Finally, assume that (P2) is violated. Since it is easy to see that each vertex in $V(G)$ will be contained in at least one image of $\chi$, this means that there would be some $v\in V(G)$ and nodes $t,t_a,t_b\in V(T)$ such that:
\begin{compactitem}
\item $v\not \in \chi(t)$ but $v\in \chi(t_a)$ and $v\in \chi(t_b)$;
\item $t$ separates $t_a$ from $t_b$ in $T$.
\end{compactitem}

If $v$ is \black, then this would immediately violate the fact that $\mathcal{T'}$ is a tree decomposition of $G'$. On the other hand, if $v$ is \red, then there would have to exist $v_i$ and $v_j$ such that $v_i\in \chi'(t_a)$ and $v_j\in \chi'(t_b)$; moreover, $v_i\neq v_j$ since otherwise we would once again contradict (P2) for $\mathcal{T'}$. But then by construction we know that $v_iv_j\in E(G')$.
Thus, by (P1) there is a bag $t' \in V(T)$ for which $v_i, v_j \in \chi'(t')$.
By (P2) there is a path in~$T$ from $t_a$ (resp.\ from $t_b$) to $t'$ on which each bag $s$ has $v_i \in \chi(s)$ (resp.\ $v_j \in \chi(s)$).
One of these paths contains $t$ and thus $v \in \chi(t)$, a contradiction.
Hence (P2) holds as well, completing the proof.
\end{proof}

Lemma~\ref{lem:approxctw} and the algorithm of \citet{Bodlaender96} can be used to approximate \ctw{c}:

\begin{theorem}
\label{thm:ctwapprox}
Given $c\in \Nat$, a \colored\ graph $G$ and $k\in \Nat$, in $(ck)^{\bigoh((ck)^3)}\cdot |V(G)|$ time it is possible to either correctly determine that the \ctw{c} of $G$ is at least $k+1$ or to output a $(ck + k)$-width \ctd{c}\ of $G$ with $\bigoh(|V(G)|)$ nodes.
\end{theorem}

\begin{proof}
  First, we construct the graph $G'$ as per Lemma~\ref{lem:approxctw}.
  By that lemma, if $G$ has \ctw{c} at most $k$, then $G'$ has treewidth at most $ck+k=\ell$.
  We then apply the fixed-parameter linear-time algorithm for treewidth of~\citet{Bodlaender96} to compute a tree decomposition of width at most $\ell$, or correctly determine that no such tree decomposition exists---in which case we output ``NO''.
  Applying this algorithm takes $\ell^{\bigoh(\ell^3)}\cdot |V(G)|$~time (see also \citet{BodlaenderDDFLP16}).
  If the output is NO, then the \ctw{c} of $G$ is at least~$k + 1$, as required.
  If a decomposition for $G'$ is found, we translate it back to $G$ using Lemma~\ref{lem:approxctw} and output the result.
  By Lemma~\ref{lem:approxctw} the treewidth of the output decomposition is at most $ck + k$ and the load is at most \[\frac{\ell}{k + 1} = \frac{c(k + 1) + k - c}{k + 1} = c + \frac{k - c}{k + 1}.\]
  Since the load is an integer, it is at most $c$, as claimed.
\end{proof}

By constructing the graph $G^d_\III$ and then computing a \ctd{c}\ of $G^d_\III$ with width at most $ck + k$ using Theorem~\ref{thm:ctwapprox}, in combination with Theorem~\ref{thm:usingctw}, we obtain:

\newcommand\solvetime{\ensuremath{d^{ck + k + 1} \cdot |\III|^{c+2} + (ck)^{\bigoh((ck)^3)} \cdot |\III|}}
\begin{theorem}
  \label{thm:ctwmain}
  Given $c, d \in \Nat$, and a \textsc{CSP} instance $\III$, %
  we can solve $\III$ in \solvetime\
  time where $k$ is the threshold\nobreakdash-$d$ \ctw{c} of $\III$.
  Thus, for constant $c$ and $d$, \textsc{CSP} is fixed-parameter tractable parameterized by $k$.
\end{theorem}
\begin{proof}
  The algorithm is as follows.
  We first construct the graph $G^d_\III$.
  Since $\III$ has threshold\nobreakdash-$d$ \ctw{c} at most~$k$, the maximum number of variables in a constraint is at most~$k + 1$.
  Thus, $G^d_\III$ can be computed in $O(k^2 \cdot |\III|)$ time by initializing an empty graph with a vertex for each variable of $\III$, marking the vertices as \red\ that correspond to variables with domain size more than~$d$, and then iterating over all constraints and adding the corresponding edges.
  Then, we compute a \ctd{c}\ of $G^d_\III$ with width at most $ck + k$ using Theorem~\ref{thm:ctwapprox}.
  This takes $ck^{\bigoh((ck)^3)} \cdot k^2 \cdot |\III|$ time.
  The result then follows from Theorem~\ref{thm:usingctw}.
\end{proof}

Note that the runtime bound stated in
Theorem~\ref{thm:ctwmain} would allow us to take the threshold $d$ as
an additional parameter instead of a constant, to still establish
fixed-parameter tractability of CSP, parameterized by $k+d$.
\smallskip

\section{Further Applications of Threshold Treewidth}\label{sec:tw-appl}
While our exposition here focuses primarily on applications for the classical constraint satisfaction problem, it is worth noting that \ctw{c} can be applied analogously on many other prominent problems that arise in the AI context. In this subsection, we outline three such applications of our machinery in highly general settings.

\subsection*{Weighted Constraint Satisfaction with Default Values}
Our first application concerns a recently introduced extension of constraint satisfaction via a combination of weights and default values~\cite{BraultbaronCapelliMengel15,GanianKimSlivovskySzeider21} (see also the published preprint by~\citet{GanianKimSlivovskySzeider18}). This extension captures, among others, counting CSP (\textsc{\#CSP}) and counting SAT (\textsc{\#SAT}).
We introduce the extension below by building on our
preliminaries on \textsc{CSP}. 

For a variable set $V$ and a domain $D$, a \emph{weighted constraint $C$ of arity $\rho$ over $D$ with default value $\eta$} (or ``weighted constraint'' in brief) is a tuple $C=(S,F,f,\eta)$
where 
\begin{itemize}
\item the \emph{scope} $S=(x_1,\dots,x_\rho)$ is a sequence of variables from $V$, 
\item $\eta\in \mathbb{Q}$ is a rational number called the \emph{default value}, 
\item $F\subseteq D^{\,\rho}$ is called the \emph{support}, and
\item $f:F \rightarrow \mathbb{Q}$ is a mapping 
which assigns rational weights to the support.
\end{itemize}

A weighted constraint $c=(S,F,f,\eta)$ naturally induces a total function on assignments of its scope $S=(x_1,\dots,x_\rho)$: for each assignment $\alpha:X\rightarrow D$ where $X\supseteq S$, we define the \emph{value}~$c(\alpha)$ of~$c$ under~$\alpha$ as $c(\alpha)=f(\alpha(x_1),\dots, \alpha(x_\rho))$ if $(\alpha(x_1),\dots, \alpha(x_\rho))\in F$ and $c(\alpha)=\eta$ otherwise. 

Similarly to \CSP, an instance $\III$ of \textsc{Weighted Constraint Satisfaction with Default Values} (\sCSPD) is a tuple $(V,D,C)$, but here $C$ is a set of weighted constraints. The task in \sCSPD\ is to compute the total weight of all assignments of $V$, i.e., to compute the value $\sol(\III)=\sum_{\alpha:V\rightarrow D} \prod_{c\in C}c(\alpha)$.

\sCSPD\ was shown to be fixed-parameter tractable when parameterized by the treewidth of the primal graph plus $|D|$~\cite{GanianKimSlivovskySzeider18}, in particular as a corollary of a more general dynamic programming algorithm $\mathbb{A}$~\cite[Theorem 1]{GanianKimSlivovskySzeider18}. When $\mathbb{A}$ is applied on the primal graph, it proceeds in a leaf-to-root fashion that is similar in nature to the algorithm described in the proof of Theorem~\ref{thm:usingctw} here; however, formally the records stored by $\mathbb{A}$ are more elaborate. In particular, at each node~$t$ of a provided tree decomposition, $\mathbb{A}$ stores one record for each pair $(\theta,\vec{B})$ where 
\begin{itemize}
\item $\theta$ is an assignment of the vertices in $\chi(t)$, and 
\item $\vec{B}$ is a tuple that specifies for each constraint that is ``processed'' at $t$ the subset of tuples in the support that agree with~$\theta$.
\end{itemize}

Crucially, when applying $\mathbb{A}$ on the primal graph, in every tuple $(\theta,\vec{B})$ the latter component $\vec{B}$ is fully determined by the former component. And since the number of possible choices for $\theta$ is upper-bounded by $d^k\cdot |\III|^{c+2}$ for the same reason as in Theorem~\ref{thm:usingctw}, we obtain:

\begin{proposition}\label{prop:application-weighted-csp}
Given $c, d \in \Nat$, and an instance $\III$ of $\sCSPD$ it is possible to solve $\III$ in $\solvetime$ time where $k$ is the threshold-$d$ \ctw{c} of~$\III$.
In particular, for constant~$c$ and $d$ $\sCSPD$ is fixed-parameter tractable parameterized by $k$.
\end{proposition}

\subsection*{Valued Constraint Satisfaction}
The second application is for the \textsc{Valued CSP} (\textsc{VCSP}) \cite{SchiexFV95,Zivny12}.
Herein, we are given the same input as in $\sCSPD$ but where every weighted constraint has a default value of $0$. The goal in \textsc{VCSP} is to compute a variable assignment $\alpha$ that minimizes $\sum_{c \in C}c(\alpha)$. \textsc{VCSP} generalizes \textsc{MaxCSP}, where we aim to find an assignment for a \textsc{CSP} instance that maximizes the number of satisfied constraints.

It is a folklore result that \textsc{VCSP} can be solved by a dynamic programming algorithm along a tree decomposition of the primal graph, yielding \XP-tractability when parameterized by the treewidth of the primal graph~\cite{Carbonnel0Z18,BerteleBrioschibook}. The algorithm can be seen as a slight extension of the one presented in Theorem~\ref{thm:usingctw}: the records $M(t)$ used in the algorithm that keep a list of all assignments $\theta$ are enhanced to also keep track of the value $\sum_{c \subseteq V_t}c(\theta)$. We thus obtain the following.

\begin{proposition}\label{prop:application-vcsp}
Given $c, d \in \Nat$ and an instance $\III$ of \textsc{VCSP} it is possible to solve $\III$ in \solvetime\ time where $k$ is the threshold\nobreakdash-$d$ \ctw{c} of~$\III$.
In particular, for constant~$c$ and $d$ \textsc{VCSP} is fixed-parameter tractable parameterized by $k$.
\end{proposition}

\subsection*{Integer Programming}
Our third application concerns \textsc{Integer Programming}
(\textsc{IP})~\cite{Schrijver99}, the generalization of the famous \textsc{Integer Linear Programming} problem to arbitrary polynomials. \textsc{IP} is, in fact, \emph{undecidable} in general; see \citet{Koppe12} for a survey on its complexity.
However, when there are explicit bounds on the variable domains, it can be solved by a fixed-parameter algorithm via dynamic programming on tree decompositions.

For our presentation, we provide a streamlined definition of \textsc{IP} with domain bounds as used, e.g., by \citet{EibenGKO19}. An instance of \textsc{IP} consists of a tuple $(X,\mathcal{F},\beta,\gamma)$ where: 
\begin{itemize}
\item $X=\{x_1,\dots,x_n\}$ is a set of variables,
\item $\mathcal{F}$ is a set of integer polynomial inequalities over variables in $X$, that is, inequalities of the form $p(x_1, \ldots, x_n) \leq 0$ where $p$ is a sparsely encoded polynomial with rational coefficients,%
\item $\beta$ is a mapping from variables in $X$ to their domain, i.e., $\beta(x)$ is the set of all integers $z$ such that $x\mapsto z$ satisfies all constraints in $\mathcal{F}$ over precisely the variable $x$ (these are often called \emph{box constraints}), and
\item $\gamma$ is an integer polynomial over variables in $X$ called the evaluation function.
\end{itemize}

The goal in \textsc{IP} is to find an assignment $\alpha$ of the variables of $\III$ which (1) satisfies all inequalities in $\mathcal{F}$ and $\beta$ while achieving the maximum value of $\gamma$.

Let $d=\max_{x\in X}|\beta(x)|$, and let the primal graph $G_\III$ of an \textsc{IP} instance $\III$ be the graph whose vertex set is $X$ and where two variables are adjacent if and only if there exists an inequality in $\mathcal{F}$ containing both variables. It is known that \textsc{IP} is fixed-parameter tractable when parameterized by $d$ plus the treewidth of $G_\III$~\cite{EibenGKO19}. The algorithm $\mathbb{B}$ used to establish this result performs leaf-to-root dynamic programming that is analogous in spirit to the procedure used in the proof of Theorem~\ref{thm:usingctw}. Herein in particular, at each node $t$ algorithm~$\mathbb{B}$ stores records which specify the most favorable ``partial evaluation'' of $\gamma$ for each possible assignment of variables in $\chi(t)$ in view of $\beta$ and $\mathcal{F}$.

Since each variable is equipped with a domain via $\beta$, we may define the graph $G_\III^d$ in an analogous way as for \textsc{CSP}. Once that is done, it is not difficult to verify that running the algorithm of \citet{EibenGKO19} on a threshold-$d$ \ctd{c}\ of $G_\III^d$ guarantees a runtime bound for solving \textsc{IP} of $d^{\bigoh(k)}\cdot |\III|^{c + 2}$.
In combination with our Theorem~\ref{thm:ctwapprox}, we conclude:

\begin{proposition}\label{prop:application-IP}
Given $c, d \in \Nat$, and an instance $\III$ of \textsc{IP} it is possible to solve $\III$ in \solvetime\ time where $k$ is the threshold-$d$ \ctw{c} of~$\III$.
In particular, for constant~$c$ and $d$ \textsc{IP} is fixed-parameter tractable parameterized by $k$.
\end{proposition}

\section{Threshold Hypertree Width}\label{section:htw}

In this section, we define threshold hypertree width for \textsc{CSP}, show how to use it to obtain fixed-parameter algorithms, and how to compute the associated decompositions.
Similar to threshold treewidth, we will first introduce an
enhancement of hypertree width for hypergraphs.
Intuitively, the running time of dynamic programs for \textsc{CSP} based on decompositions of the corresponding hypergraph is strongly influenced by constraints, corresponding to hyperedges, whose relations contain many tuples.
We hence aim to distinguish these hyperedges%
.

\looseness=-1
Let $H$ be a hypergraph where $E = E(H)$ is bipartitioned into a set~$E_B$ of \emph{\black} hyperedges and a set~$E_R$ of \emph{\red} hyperedges.
We call such hypergraphs \emph{\colored}.
Let $c \in \Nat_0$.
A \emph{\chtd{c}} of $H$ is a hypertree decomposition $(T, \chi, \lambda)$ for $H$ such that each edge cover $\lambda(v)$, $v \in V(T)$, contains at most $c$ \red\ hyperedges.
The width and the notion of \emph{\chtw{c}} (of $H$) are defined in the same way as for hypertree decompositions.

Similar to threshold treewidth, for each fixed $c$ there are
hypergraphs that do not admit a \chtd{c}. For example, consider a clique graph with at least $c + 2$ vertices with \red\ edges only, interpreted as a hypergraph.
As a \chtd{c}\ contains a tree decomposition for the clique, there is a bag containing all vertices of this clique, and the minimum edge cover for this bag has size~$c + 1$.

\looseness=-1
We now apply the above notions to \textsc{CSP}.
Let $d,c\in \Nat_0$ and $\III=(V,D,C)$ be a \textsc{CSP} instance.
Let $H^d_\III$ be the \colored\ hypergraph of $\III$ wherein a hyperedge $F \in E(H^d_\III)$ is \black\ if and only if $|R(\gamma)|\leq d$, for the constraint~$\gamma \in C$ corresponding to $F$, i.e., $S(\gamma) = F$.
Then, the \emph{threshold-$d$ \chtw{c}} of $\III$ is the
\chtw{c}\ of~$H^d_\III$.
For threshold-$d$ \chtw{c}, we also obtain a fixed-parameter algorithm for CSP.
Instead of building on hypertree decompositions in the above, we may also use generalized hypertree decompositions, leading to the notion of \emph{generalized threshold-$d$ \chtw{c}} and the associated decompositions.

\begin{theorem}\label{thm:usingchtw}
  Given $c, d\in \Nat$, a \textsc{CSP} instance $\III$ with (generalized) threshold-$d$ \chtw{c}~$k$ together with the associated decomposition of~$H^d_\III$, in $\bigoh(d^k\cdot |\III|^{c+2})$ time it is possible to decide~$\III$ and produce a solution if there is one.
\end{theorem}
In particular, for fixed $c, d$, \textsc{CSP} is fixed-parameter tractable parameterized by~$k$ when a threshold-$d$ \chtd{c}\ of width~$k$ is given.
\begin{proof}[Proof Sketch]
  A usual approach used for ordinary hypertree decompositions is to compute an equivalent CSP whose hypergraph is acyclic and then use an algorithm for acyclic CSPs \cite{GottlobLeoneScarcello02}.
  We instead apply a direct dynamic programming approach; the stated running-time bound then follows from the upper bound on constraints with large number of tuples imposed by the definition of \chtw{c}.

  Let $(T, \chi, \lambda)$ be the \chtd{c}\ of $H^d_\III$ provided in the input.
  Root~$T$ arbitrarily and denote the root by~$r$.
  For each $t \in V(T)$, let $V_t = \bigcup_{t'} \chi(t')$, where the union is taken over all~$t'$ in the subtree of~$T$ rooted at~$t$.
  A \emph{$t$\nobreakdash-mapping} is a mapping that assigns to each variable~$v \in \chi(t)$ a value from~$D(v)$.
  
  The algorithm proceeds by dynamic programming, i.e., computing, for each node $t \in V(T)$ in a leaf-to-root fashion, the set~$M(t)$ of all $t$-mappings $\theta$ with the following two properties: (1), there exists some extension $\theta'$ of $\theta$ to $V_t$ which maps each variable $v \in V_t$ to an element of $D(v)$ such that each constraint~$\gamma$ with $S(\gamma) \subseteq V_t$ is satisfied by $\theta'$ and, (2), for each constraint $\gamma \in \lambda(t)$, mapping~$\theta$ projected\footnote{The projection of a relation~$R$ onto a subset $S$ of its variables is the set resulting from taking each tuple of $R$ and removing from this tuple the entries for variables not in~$S$.} onto $S(\gamma)$ occurs as a tuple in $\gamma$~projected onto $\chi(t)$.
  
  Observe that $\III$ is a YES-instance if and only if $M(r) \neq \emptyset$: The backward direction follows from property~(1).
  To see the forward direction, note that any satisfying assignment projected onto $\chi(r)$ is contained in $M(r)$.
  Thus, to decide $\III$ it suffices to compute all sets $M(t)$, $t \in V(T)$.
  The solution, if it exists, can then be computed by retracing the steps of the dynamic program in a standard fashion.

  Before we explain how to compute $M(t)$, consider the following way of constructing a $t$-mapping~$\theta$.
  For each constraint in $\lambda(t)$, pick a tuple such that each pair of picked tuples agree on the variables they share (if any).
  Note that the picked tuples induce a $t$-mapping, and we set $\theta$ to be this mapping.
  Call a $t$-mapping constructed in this way \emph{derived}.
  Note that the number of derived $t$-mappings is at most $d^{k - c} \cdot |\III|^c$ and that the set of all derived $t$-mappings can be computed in $\bigoh(d^{k - c} \cdot |\III|^{c + 1})$ time.

  Next, we explain how to compute~$M(t)$.
  To compute $M(\ell)$ for a leaf~$\ell$, due to property~(2), it suffices to loop over all derived $\ell$-mappings and to put them into $M(\ell)$ if they satisfy all constraints~$\gamma$ for which $S(\gamma) \subseteq \chi(\ell)$.
  By the bound on the number of derived $\ell$-mappings, this takes $\bigoh(d^k \cdot |\III|^{c + 1})$~time.

  Consider an internal node~$t$ of~$T$.
  Again, we loop over each derived $t$-mapping $\theta$ and check whether it satisfies all constraints whose scope is in $\chi(t)$.
  If not, then we discard $\theta$.
  If yes, then for each child $t'$ of $t$ we check whether there is a mapping $\theta' \in M(t')$ such that $\theta$ and $\theta'$ agree on their shared variables; in formulas $\forall v\in \chi(t)\cap \chi(t'): \theta'(v)=\theta(v)$.
  If so, then we put $\theta$ into $M(t)$.
  By using property~(2) of the mappings in $M(t')$, in this way, we correctly compute~$M(t)$. 
  Using suitable data structures and the bound on the number of derived mappings, this computation can be carried out in time at most $\bigoh(d^k \cdot |\III|^{c + 1})$ per node in~$T$.
\end{proof}

Similar to weighted treewidth, a weighted variant of hypertreewidth has been proposed~\cite{ScarcelloGrecoLeone07} wherein the whole decomposition~$(T, \chi, \lambda)$ is weighted according to the estimated running time of running a dynamic program similar to the above.
The approach is, slightly simplified, to weigh each hyperedge in the cover of a bag by $|R(c)|$ for the corresponding constraint~$c$ and then to minimize $\sum_{t \in V(T)} \Pi_{c \in \lambda(t)}|R(c)|$.
A drawback here again is that, using this quantity as a parameter, it implicitly bounds the number of tuples in each constraint $|R(c)|$ and in turn all domain sizes.
This is not the case for threshold-$d$ \chtw{c}.

We now turn to computing the decomposition for the hypergraph of the CSP used in Theorem~\ref{thm:usingchtw}.
A previous approach for computing ordinary hypertree decompositions of width at most~$k$ by first recursively decomposing the input hypergraph via separators which roughly correspond to the vertex sets of the potential covers of the bags, that is, sets~$S$ of at most $k$ hyperedges.
The decomposition can then be determined in a bottom-up fashion~\cite{Gottlob99}.
This approach can be adapted to \chtd{c}\ by replacing the sets~$S$ with sets of at most~$k$ hyperedges among which there are at most~$c$ \red\ hyperedges.
We omit the details.
Indeed, we may instead use a more general framework, due to \citet{ScarcelloGrecoLeone07}, which allows to compute hypertree decompositions of width at most~$k$ that additionally optimize an arbitrary weight function.
Applying this framework leads to the following.

\begin{theorem}\label{thm:computechtw}
  Given $c, k \in \Nat$, and a \colored\ hypergraph~$H$, in $\bigoh(|E(H)|^{2k} \cdot |V(H)|^2)$ time it is possible to compute a \chtd{c}\ for $H$ of width at most~$k$ or correctly report that no such decomposition exists.
\end{theorem}
 
\begin{proof}
  We first state the result of \citet{ScarcelloGrecoLeone07} in a simplified and weaker form that is sufficient for our purpose.
  Let $p$ be a function that assigns an integer to a bag of any hypertree decomposition.
  Let $r_p$ be the function of the running time needed to evaluate~$p$.
  A \emph{tree aggregation function} is a function that assigns to each hypertree decomposition $(T, \chi, \lambda)$ the integer $\max_{t \in V(T)}p(t)$.
  \citeauthor{ScarcelloGrecoLeone07}'s Theorems 4.4 and 4.5 now imply the following.
  There is an algorithm that, given an integer~$k$, a hypergraph~$H$, and a tree aggregation function~$f$, computes a width-$k$ hypertree decomposition for~$H$ that minimizes~$f$, or correctly decides that no such decomposition exists.
  The algorithm runs in $\bigoh(|E(H)|^{2k} \cdot |V(H)| \cdot (|V(H)| + r_p))$ time.\footnote{The running time bound follows from the analysis given by \citet{ScarcelloGrecoLeone07} in the proof of Theorem~4.5.}

  To apply this result to our setting, we put $p$ to be the function that assigns to each bag $t$ the number of \red\ hyperedges in the edge cover $\lambda(t)$.
  Thus, \citeauthor{ScarcelloGrecoLeone07}'s algorithm will compute the smallest $c$ such that there is a \chtd{c}.
  Note that $|\lambda(t)| \leq |V(H)|$ and hence $r_p = \bigoh(|V(H)|)$.
  This implies the running-time bound.
\end{proof}

\noindent Assuming FPT${}\neq{}$\W[2] the running time in Theorem~\ref{thm:computechtw} cannot be improved to a fixed-parameter tractable one, even if $c$ is constant.
This follows from the fact that the special case of deciding whether a given hypergraph without \red\ hyperedges admits a \chtd{0}\ of width at most~$k$ is \W[2]-hard with respect to~$k$ \cite{GottlobGroheMusliuSamerScarcello05}.

Bounding the threshold treewidth or threshold hypertree width of a CSP
instance constitutes a hybrid restriction and not a structural
restriction \cite{CarbonnelCooper15}, as these restrictions are
formulated in terms of the loaded primal graphs and the loaded
hypergraphs, and not in terms of the plain, unlabeled (hyper)graphs.
However, as the loaded (hyper)graphs carry only very little additional
information, we would like to label such restrictions as
\emph{semi-structural}.

\section{Elimination Orderings}%
\label{sec:elim-order}
The algorithms used in our experiments rely on a characterization of treewidth and generalized hypertree width by so-called elimination orders.
An \emph{elimination ordering} $\prec$ of a graph $G$ is a total ordering $\prec$ of $V(G)$. Let us denote the $i$-th vertex in $\prec$ as $v_i$, and let $G_0=G$. For each $i\in [|V(G)|]$, let the graph $G_i$ be obtained from $G_{i-1}$ by removing $v_i$ and adding edges between each pair of vertices in the neighborhood of $v_i$ (i.e., the neighborhood, $N_{G_{i - 1}}(v_i)$, of $v_i$ in $G_{i - 1}$ becomes a clique in~$G_i$). The \emph{width} of $v_i$ w.r.t.\ $\prec$ is then defined as $|N_{G_{i - 1}}(v_i)|$, and the \emph{width} of $\prec$ is the maximum width over all vertices in $G$ w.r.t.\ $\prec$.

It is well known that a graph $G$ has treewidth $k$ if and only if it admits an elimination ordering $\prec$ of width $k$~\cite{Kloks94,BodlaenderK10}.
Moreover, a tree decomposition of width $k$ can be computed from such $\prec$ and, vice-versa, given a tree decomposition of width $k$ one can construct a width-$k$ elimination ordering in polynomial time~\cite{Kloks94,BodlaenderK10}.

Recently, it has been shown that generalized hypertree decompositions
of CSP instances can be characterized in a similar
way~\cite{FichteHecherLodhaSzeider18}. In particular, consider a CSP
instance $\III$ with primal graph $G_\III$ and an elimination ordering
$\prec$ of~$G_\III$. The \emph{cover width} of $v_i$ w.r.t.\
$\prec$ is then defined as the size of a minimum edge cover of
$N_{G_{i - 1}}(v_i) \cup \{v_i\}$ in $H_\III$, and the \emph{cover width} of
$\prec$ is the maximum cover width over all vertices in $G$
w.r.t.\ $\prec$. Analogously as in the treewidth case, a generalized
hypertree decomposition of width $k$ can be computed from an
elimination ordering $\prec$ of cover width $k$, and, vice-versa, given a generalized hypertree decomposition of width $k$ one can construct a cover width-$k$ elimination ordering in polynomial time~\cite{FichteHecherLodhaSzeider18,SchidlerSzeider20}.

It is relatively straightforward to adapt these notions of elimination orderings to describe not only classical treewidth and generalized hypertree width, but also their threshold variants.
In particular, by simply retracing the steps of the original proofs~\cite{Kloks94,FichteHecherLodhaSzeider18}, one can show the following.
Recall that for a \CSP\ instance $\III$ and an integer~$d$, we have defined $G^d_\III$ as the loaded graph obtained from the primal graph~$G_\III$ of $\III$ by marking each vertex $v \in V(G_\III)$ as \black\ if $|D(v)| \leq d$ and \red\ otherwise.
Also, $H^d_\III$ is the loaded hypergraph obtained from the hypergraph~$H_\III$ of $\III$ wherein we mark each hyperedge $F \in E(H_\III)$ as \black\ if $R(\gamma) \leq d$, where $\gamma$ is the constraint corresponding to $F$, and we mark $F$ as \red\ otherwise.
\begin{theorem}\label{the:orderings}
  \looseness=-1
  \textnormal{(1)} A \CSP\ instance $\III$ has threshold-$d$ \ctw{c} $k$ if and only if $G_\III^d$ admits an elimination ordering of width $k$ with the property that for each $v_i$, $N_{G_{\III,i - 1}^d}(v_i) \cup \{v_i\}$ contains at most $c$ \red\ vertices.
  \textnormal{(2)} A \CSP\ instance $\III$ has generalized threshold-$d$ \chtw{c} $k$ if and only if $G_\III$ admits an elimination ordering of cover width~$k$ with the property that for each $v_i$, $N_{G_{\III,i - 1}^d}(v_i) \cup \{v_i\}$ admits a hyperedge cover (in $H^d_\III$) of size at most $k$ containing at most $c$ \red\ hyperedges.
\end{theorem}
\begin{proof}
  We prove both parts of the statement simultaneously; we mainly describe the proof of part~\textrm{(1)} and while doing so explain the differences to obtain part~\textrm{(2)}.
  First, we show the direction from a tree decomposition (resp.\ hypertree decomposition) to an elimination ordering.
  Let $\III$ be a \CSP\ instance with threshold-$d$ \ctw{c}~$k$ (resp.\ with generalized threshold-$d$ \chtw{c}~$k$).
  Let $(T, \chi)$ be a load-$c$ tree decomposition of width~$k$ for $G^d_\III$.
  For the case of hypertree width, let $(T, \chi, \lambda)$ be a generalized \chtd{c}.
  Let $n := |V(G^d_\III)|$.
  Proceed as follows:
  Put $G'_0 = G^d_\III$ and $\mathcal{T}_0 = (T_0, \chi_0) := (T, \chi)$; respectively, put $\mathcal{T}_0 = (T_0, \chi_0, \lambda_0) := (T, \chi, \lambda)$.
  Then, for each $i = 1, 2, \ldots, n$ construct a graph $G'_i$, a vertex~$v_i$, and a tree decomposition $\mathcal{T}_i = (T_i, \chi_i)$ (resp.\ a generalized hypertree decomposition $\mathcal{T}_i = (T_i, \chi_i, \lambda_i)$) as follows.
  Herein, we maintain the invariant that $\mathcal{T}_i$ is a \ctd{c} of width~$k$ for~$G'_i$ (resp.\ a \chtd{c} of width~$k$ for $H_i$, the hypergraph obtained from $H$ by removing $v_1, \ldots, v_{i - 1}$).  
  \begin{enumerate}
  \item\label{the:orderings:step:leaves}
    Pick an arbitrary leaf~$t$ in~$T_{i - 1}$.
    If each vertex in $\chi_{i - 1}(t)$ occurs in the parent of~$t$ in~$T_{i - 1}$, remove $t$ from $T_{i - 1}$.
    Note that this results in another (generalized hyper-) tree decomposition of at most the same width and load.
    If $t$ was removed, repeat this step.
  \item After Step~\ref{the:orderings:step:leaves}, in the picked leaf~$t \in \mathcal{T}$ there is a vertex~$v \in \chi_{i - 1}(t)$ that occurs in no other bag of~$\mathcal{T}_{i - 1}$.
    Put $v_i := v$.
  \item\label{the:orderings:step:construct} To obtain $G'_{i}$, take $G'_{i - 1}$, remove $v_i$, and make $N_{G'_{i - 1}}(v_i)$ a clique.
    To obtain $\mathcal{T}_{i}$, take $\mathcal{T}_{i - 1}$ and remove $v_i$ from all bags.
    Observe that this maintains our invariant because $N_{G_{i - 1}}(v_i)$ is contained in the bag $\chi_{i - 1}(t)$.
  \end{enumerate}
  We claim that the elimination ordering $\prec$ on $V(G^d_\III)$ induced by $v_1, v_2, \ldots, v_n$ has (cover) width $k$ and for each $v_i$ we have that $N_{G^d_{\III, i -1}}(v_i) \cup \{v_i\}$ contains at most $c$ heavy vertices (resp.\ for each $v_i$ we have that $N_{G^d_{\III, i -1}}(v_i) \cup \{v_i\}$ admits a hyperedge cover in $H_\III$ of size at most~$k$ and with at most $c$ heavy hyperedges).
  Indeed, $G'_i$ is equal to the graph $G_i$ defined by $\prec$.
  In the case of tree decompositions, since $N_{G'_{i - 1}}(v_i) \cup \{v_i\}$ is contained in the bag $\chi_{i - 1}(t)$ in Step~\ref{the:orderings:step:construct} and since $\mathcal{T}_{i - 1}$ is a width-$k$ load-$c$ tree decomposition for~$G'_{i - 1}$, the ordering~$\prec$ has width $k$ and there are at most $c$ heavy vertices in $N_{G_{i - 1}}(v_i) \cup \{v_i\}$.
  Similarly, in the case of hypertree decompositions, since $N_{G'_{i - 1}}(v_i) \cup \{v_i\}$ is contained in the bag $\chi_{i - 1}(t)$ in Step~\ref{the:orderings:step:construct} and since $\mathcal{T}_{i - 1}$ is a width~$k$ load-$c$ generalized hypertree decomposition for~$G'_{i - 1}$, the required cover of $N_{G'_{i - 1}}(v_i)$ is given by $\lambda_{i - 1}(t)$.
  Thus, the ordering~$\prec$ has cover width $k$ and $N_{G'_{i - 1}}(v_i) \cup \{v_i\}$ admits a cover of size at most~$k$ with at most $c$ heavy hyperedges.
  This completes the argument for the direction from tree decompositions to elimination orderings.

  Now let $\prec$ be an elimination ordering for $G^d_\III$ with the properties promised in part~\textrm{(1)} of the theorem (resp.\ in part~\textrm{(2)}).
  Let $v_1, v_2, \ldots, v_n$ be the ordering of vertices of $G_\III$ induced by $\prec$ and let $G_1, G_2, \ldots, G_n$ be the corresponding graphs.
  Let $G_{n + 1}$ be the empty graph and let $\mathcal{T}_{n + 1} = (T_{n + 1}, \chi_{n + 1})$ be a trivial tree decomposition for~$G_{n + 1}$ wherein $T_{n + 1}$ consists of a single vertex~$t$ and the corresponding bag is empty.
  For hypertree decompositions we let $\mathcal{T}_{n + 1} = (T_{n + 1}, \chi_{n + 1}, \lambda_{n + 1})$ be an analogous hypertree decomposition, where additionally $\lambda_{n + 1}(t) = \emptyset$.
  For each $i = n, n - 1, \ldots, 1$ we construct a tree decomposition $\mathcal{T}_i = (T_i, \chi_i)$ for $G_i$ (resp.\ a generalized hypertree decomposition for $H_i$, the hypergraph obtained from $H_\III$ by removing the vertices $v_1, v_2, \ldots, v_{i - 1}$). 
  Herein, we maintain the invariant that $\mathcal{T}_i$ is a (generalized hyper-) tree decomposition for $G_i$ (resp.\ $H_i$) of width at most~$k$ and load at most~$c$.
  At Step~$i$, proceed as follows.
  Take $\mathcal{T}_{i + 1}$ and find a node~$t$ in $T_{i + 1}$ such that the bag $\chi_{i + 1}(t)$ contains $N_{G'_{i}}(v_i)$.
  Such a node exists, because $N_{G'_{i}}(v_i)$ is a clique in~$G_{i + 1}$.
  To obtain $\mathcal{T}_i$ from $\mathcal{T}_{i + 1}$, add a new vertex~$t'$ as a child of $t$ to $T_{i + 1}$ and define $\chi_i(t') := N_{G'_{i}}(v_i) \cup \{v_i\}$.
  Since $\prec$ has width $k$ and there are at most $c$ \red\ vertices in $N_{G'_{i}}(v_i) \cup \{v_i\}$, we have that $\mathcal{T}_i$ is a \ctd{c} of width at most~$k$ for $G_i$.
  For hypertree decompositions, define also $\lambda_{i}(t')$ as the hyperedge cover in $H_\III$ of $N_{G'_{i}}(v_i) \cup \{v_i\}$ that has size at most~$k$ and contains at most~$c$ \red\ hyperedges.
  Since $H_i$ is a subhypergraph of $H_\III$, this cover is also a cover in~$H_i$.
  Thus, $\mathcal{T}_i$ is a \chtd{c} of width at most $k$ for $H_i$.
  This finishes the proof.
\end{proof}

\looseness=-1 A (significantly more complicated) elimination ordering
characterizations of hypertree width has been obtained by
\citet{SchidlerSzeider20,SchidlerSzeider21b}. These, too, can be
translated into characterizations of threshold-$d$ \chtw{c}. However,
experimental evaluations confirmed the expectation that there was no
practical benefit to using hypertree width instead of generalized
hypertree width.

\section{Implemented Algorithms}
\label{sec:practical}

\begin{table}[t]
  \caption{Overview of the algorithms that we use. The acronym \ghtws\ stands for generalized hypertree width.
    We use \emph{small} to indicate that the corresponding quantity is optimized using heuristic methods.
    We use \emph{second} to indicate that the corresponding quantity was optimized as a secondary objective.}
  \label{tab:overview-algos}
    \centering
    \begin{tabular}{@{}lllp{6.5cm}@{}}
      \toprule
      Parameter & Type  & Name     & Description\\
      \midrule
      treewidth & Exact & \algtwxo & Minimum width, disregarding load.\\
                &       & \algtwxw & Minimum width, load second.\\
                &       & \algtwxl & Minimum load, treewidth second.\\
                & Heuristic & \algtwho & Small width, disregarding load.\\
                &           & \algtwhw & Small width, load second.\\
                &            & \algtwhl & Small load, treewidth second.\\
      \ghtws     & Exact     & \alghtxo & Minimum width, disregarding load.\\
                &           & \alghtxw & Minimum width, load second.\\
                &           & \alghtxl & Minimum load, width second.\\
                & Branch \& Bound & \alghtho & Minimum (cover) width for heuristic tree decomposition, disregarding load.\\
                &                  & \alghthw & Minimum (cover) width for heuristic tree decomposition, load second.\\
                &                 & \alghthl & Minimum load for heuristic tree decomposition, (cover) width second.\\
                & Greedy          & \alghtgo & Small width, disregarding load.\\
                &                 & \alghtgw & Small width, load second.\\
    \bottomrule
    \end{tabular}
\end{table}

We use classical exact and heuristic algorithms to compute tree decompositions and generalized hypertree decompositions and adapt them to take the load into account as described below.
We call the algorithms without adaptions (load-) \emph{oblivious}.
These algorithms will bear the suffix \emph{Obl} in the identifiers for the implemented algorithms that we introduce below.
The adapted algorithms either minimize the width of the decomposition first (with heuristic or exact methods) and the load second, represented by suffix \emph{W$\rightarrow$L}, or load first and width second, represented by suffix \emph{L$\rightarrow$W}.
Algorithms for treewidth are prefixed with \emph{TW} and algorithms for (generalized) hypertree width are prefixed with \emph{HT}.
An overview over all algorithms can be found in Table~\ref{tab:overview-algos}.

All our algorithms are based on elimination orderings.
A minimum-width elimination ordering without taking heavy vertices into account for a given graph can be computed
using a SAT encoding~\cite{SamerV09}; below we call this algorithm \algtwxo.
This encoding can be extended to
compute optimal generalized hypertree decompositions, by computing the covers for a tree decomposition of the primal graph~\cite{FichteHecherLodhaSzeider18}
using an SMT encoding, below denoted by \alghtxo.
The SMT approach is highly robust and can be
adapted to also compute threshold-$d$ \ctd{c}s: analogously to the existing cardinality constraints
for bags/covers, we add new constraints that limit the number of heavy
vertices/hyperedges (see Theorem~\ref{the:orderings}).
We use the SMT approach to either compute a decomposition that minimizes the width first and the load second, that is, a decomposition that has minimum width and, among all decompositions with minimum width, minimum load (leading to algorithms \algtwxw\ and \alghtxw).
Or we use the SMT approach to compute a decomposition that minimizes the load first and the width second, that is, a decomposition that has minimum load and, among all decompositions with minimum load, minimum width (leading to algorithms \algtwxl\ and \alghtxl).

Since optimal elimination orderings of graphs are hard to compute, heuristics are
often used. The \emph{min-degree heuristic} constructs an ordering in a
greedy fashion by choosing the $i$-th vertex, $v_i$, in the ordering among the vertices of minimum degree
in the graph $G_{i-1}$ as defined above, and yields decompositions with good width values overall~\cite{BodlaenderK11}.
Below we call this algorithm \algtwho.
We adapted this method into two new heuristics that consider load: \algtwhl\ and \algtwhw.
The former chooses all the heavy vertices first; that is, it selects the $i$-th vertex, $v_i$, in the ordering as an arbitrary heavy vertex in $G_{i - 1}$ of minimum degree or, if $G_{i - 1}$ does not contain any heavy vertices, then it selects $v_i$ to be an arbitrary vertex in $G_{i - 1}$ of minimum degree.
This leads to decompositions with low load but possibly larger width.
The latter heuristic (\algtwhw) maintains a bound $\ell \in \mathbb{N}$ on the target load of the decomposition,
and selects the $i$-th vertex $v_i$ in the ordering as an arbitrary vertex of minimum degree among all vertices in $G_{i - 1}$ that have at most $\ell$ heavy neighbors in $G_{i-1}$; if no such vertex exists, the heuristic restarts with an incremented value of~$\ell$.

Our heuristics for generalized hypertree width follow the general framework introduced by \citet{Dermaku2008}.
In particular, they begin by computing an elimination ordering for the primal graph using the min-degree heuristic,
and then compute an edge cover for each bag.
We use the same approach and employ two different methods to compute the covers: greedy and branch \& bound (b\&b).

The branch \& bound heuristic computes an optimal edge cover for each bag.
Although this approach optimally solves an in general NP-hard problem, it is viable in our data since the resulting \textsc{Set Cover} instances are comparatively easy.
For convenience, let us call the size of the edge cover also its \emph{width} and let the \emph{load} of an edge cover be the number of heavy hyperedges contained in the cover.
Note that minimizing the width (resp.\ load) of the cover corresponds to minimizing the width (resp.\ load) of the resulting decomposition. 
We use three different objectives: minimize the width of the cover only (\alghtho), minimize width first and load second (\alghthw), and minimize load first and width second (\alghthl).

The greedy heuristic is a faster alternative to the branch \& bound approach.
The oblivious algorithm (\alghtgo) always adds the hyperedge that covers the most uncovered vertices of the current bag.
Recall that this results in covers of width at most $(1 + \log n)$ times the minimum width of a cover, where $n$ is the number of vertices (see, e.g., \citet{chvatal_greedy_1979} or Theorem~1.11 by~\citet{williamson_design_2011}).
We take the load into account by using the number of heavy hyperedges as a tie breaker when choosing the hyperedges (\alghtgw).
This corresponds to a width first and load second strategy.

\section{Experiments}\label{sec:exp}
\begin{figure*}[t]%
 \begin{subfigure}[t]{.250\textwidth}%
   \centering%
     \caption{load/exact}%
   \includegraphics[width=.99\linewidth, trim={0 0 0 0}, clip]{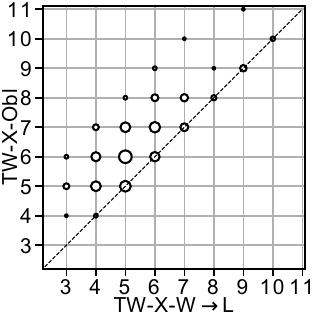}%
   \label{fig:exp_ctw_ccomp_a}%
 \end{subfigure}
 \begin{subfigure}[t]{.255\textwidth}%
   \centering%
   \caption{width/exact}%
\includegraphics[width=.99\linewidth, trim={0 0 0 0}, clip]{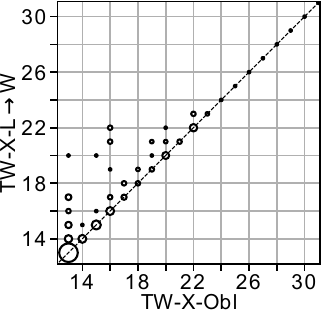}%
   \label{fig:exp_ctw_ccomp_b}%
 \end{subfigure}
 \begin{subfigure}[t]{.250\textwidth}
   \centering%
   \caption{load/exact}%
   \includegraphics[width=.99\linewidth, trim={0 0 0 0}, clip]{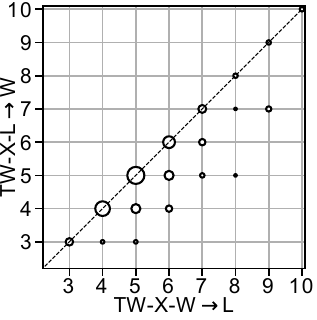}%
   \label{fig:exp_ctw_ccomp_c}%
 \end{subfigure}
 \centering
\begin{subfigure}[t]{.310\textwidth}%
   \centering%
     \caption{load/heuristic}%
   \includegraphics[width=0.9\linewidth, trim={0 0 0 0}, clip]{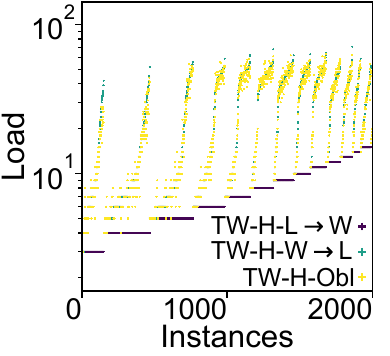}%
 	\label{fig:exp_ctw_hcomp_a}%
 \end{subfigure}
 \begin{subfigure}[t]{.310\textwidth}%
   \centering%
     \caption{width/heuristic}%
 \includegraphics[width=0.9\linewidth, trim={0 0 0 0}, clip]{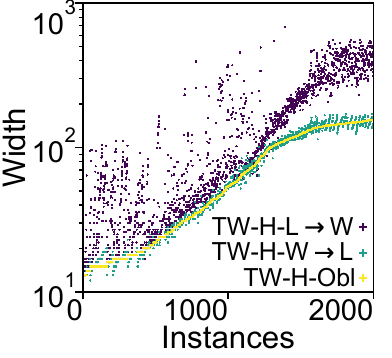}%
   \label{fig:exp_ctw_hcomp_b}%
 \end{subfigure}

   \caption{Exact and heuristic computations of tree decompositions: differences in values depending on the optimization strategy.}%
 \label{fig:exp_ctw_ccomp}%
\end{figure*}
In this section we present experimental results using the algorithms discussed in the previous section. 
We were particularly interested in the difference in loads between oblivious (\texttt{Obl}) and width-first load-second (\texttt{W$\shortrightarrow$L}) methods, and the trade-off between width-first (\texttt{W$\shortrightarrow$L}) and load-first (\texttt{L$\shortrightarrow$W}) methods.

\paragraph{Setup}
\looseness=-1
We ran our experiments on a cluster, where each node consists of two Xeon E5-2640 CPUs, each running 10 cores at 2.4\,GHz and 160\,GB memory.
As solvers for the SAT and SMT instances we used \emph{minisat 2.2.0}~\citep{Een04}\footnote{\url{http://minisat.se/}}
and \emph{optimathsat 1.6.2}~\cite{SebastianiT20}\footnote{\url{http://optimathsat.disi.unitn.it/}}. %
The control code and heuristics use \emph{Python 3.8.0}.
Our code is freely available.\footnote{See \url{https://github.com/ASchidler/htdsmt/tree/weighted} and \url{https://github.com/ASchidler/tw-sv}.}
The nodes run \emph{Ubuntu 18.04}.
We used a 8\,GB memory limit and a 2 hour time limit per instance.

\paragraph{Instances}
For threshold-$d$ \ctd{c}s\ we used 2788 instances from the
\emph{twlib}\footnote{\url{http://www.cs.uu.nl/research/projects/treewidthlib/}} benchmark set.
For generalized
threshold-$d$ \chtd{c}s\ we used the 3071 hyperbench~\cite{Fischl19}\footnote{\url{http://hyperbench.dbai.tuwien.ac.at/}}
instances after removing self-loops and subsumed hyperedges. We
created our loaded instances by marking a certain percentage of all
vertices or hyperedges as heavy. We ran experiments for different
ratios, but since the outcomes did not deviate too much, here we only present the
results for a ratio of 30\% heavy vertices/hyperedges (same as
by \citet{KaskGelfandDechter11}). 

Since instances of low width are considered efficiently solvable, our
presentation only focuses on high-width instances. In particular, for
treewidth and generalized hypertree width, we disregarded instances of width below 13 and below 4, respectively.
We were not able to find solutions for all instances; the number of instances with solutions is stated below.

\paragraph{Plots} We use a specific type of scatter plot: the position of the marker shows the pairs of values of the data point, while the size of the marker shows the number of instances for which these values were obtained. The measured quantities are noted in the plot caption.
For example, the data points in Figure~\ref{fig:exp_ctw_ccomp}a are, for each of the solved instances, the pair of loads of the tree decompositions computed by the \algtwxw\ and \algtwxo\ methods from Section~\ref{sec:practical}.

\paragraph{Treewidth}  
Figures~\ref{fig:exp_ctw_ccomp}a to~\ref{fig:exp_ctw_ccomp}c show the results from running the exact algorithms (methods \texttt{TW-X}; 168 instances could be solved within the time limit). It %
shows that even by using \texttt{W$\shortrightarrow$L} methods, we can significantly improve the load without increasing the width. Further improvements in load can be obtained by using \algtwxl, as seen in Figure~\ref{fig:exp_ctw_ccomp}c. In Figure~\ref{fig:exp_ctw_ccomp}b we see that the trade-off (in terms of the width) required to achieve the optimal loads is often very small.

The results are different for heuristic methods.
Figures~\ref{fig:exp_ctw_ccomp}d and~\ref{fig:exp_ctw_ccomp}e show the results from the 2203 instances with high width. While good estimates for load or width are possible, finding good estimates for both at the same time is not possible with the discussed heuristics:
In Figure~\ref{fig:exp_ctw_ccomp}d we see that both the \algtwho\ and \algtwhw\ heuristics mostly fail to find a good estimate for the load.
On the other hand, Figure~\ref{fig:exp_ctw_ccomp}e shows that \algtwhl\ tends to result in decompositions with much larger width than the optimum.
These results suggest that it may be non-trivial to obtain heuristics which provide a good trade-off between load and width. 

\paragraph{Generalized hypertree width}
\begin{figure}[t]
\centering
 \begin{subfigure}{.25\textwidth}
   \centering
     \caption{load/exact}
   \includegraphics[width=0.95\linewidth, trim={0 0 0 0}, clip]{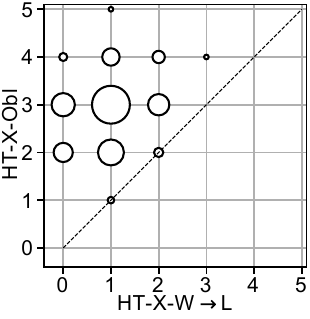}
   	\label{fig:exp_chtw_ccomp_a}
 \end{subfigure}%
 \begin{subfigure}{.25\textwidth}
   \centering
   \caption{width/exact}
   \includegraphics[width=0.95\linewidth, trim={0 0 0 0}, clip]{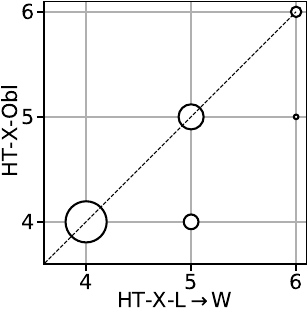}
   	\label{fig:exp_chtw_ccomp_b}
 \end{subfigure}%
  \begin{subfigure}{.25\textwidth}
   \centering
       \caption{load/exact}
   \includegraphics[width=0.95\linewidth, trim={0 0 0 0}, clip]{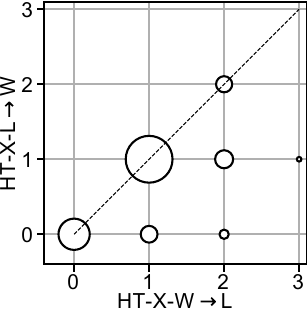}
    \label{fig:exp_chtw_ccomp_c}
 \end{subfigure}%
 \begin{subfigure}{.25\textwidth}
   \centering
     \caption{load/b\&b}
   \label{fig:exp_chtw_bcomp_a}
   \includegraphics[width=0.95\linewidth, trim={0 0 0 0}, clip]{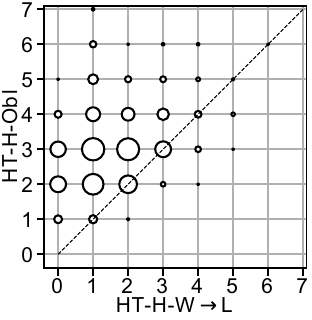}

 \end{subfigure}%

\begin{subfigure}{.25\textwidth}%
     \centering
     \caption{width/b\&b}
     \label{fig:exp_chtw_bcomp_b}
    \includegraphics[width=0.95\linewidth, trim={0 0 0 0}, clip]{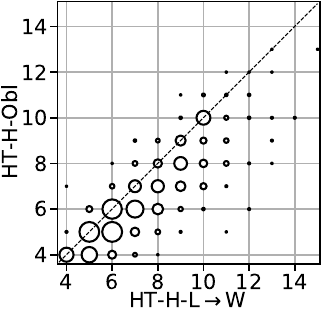}%
 \end{subfigure}%
\begin{subfigure}{.25\textwidth}%
     \centering
     \caption{load/b\&b}
     \label{fig:exp_chtw_bcomp_c}
    \includegraphics[width=0.95\linewidth, trim={0 0 0 0}, clip]{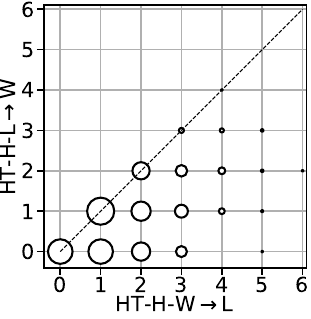}%
 \end{subfigure}%
\begin{subfigure}{.25\textwidth}%
     \centering
     \caption{load/greedy}
     \label{fig:exp_chtw_gcomp_a}
    \includegraphics[width=0.95\linewidth, trim={0 0 0 0}, clip]{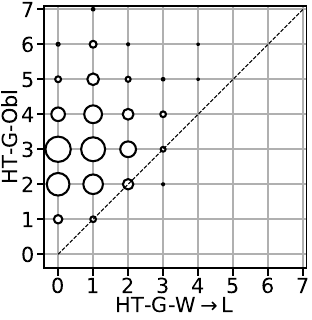}%
 \end{subfigure}%
\begin{subfigure}{.25\textwidth}%
     \centering
     \caption{width/greedy}
     \label{fig:exp_chtw_gcomp_c}
    \includegraphics[width=0.95\linewidth, trim={0 0 0 0}, clip]{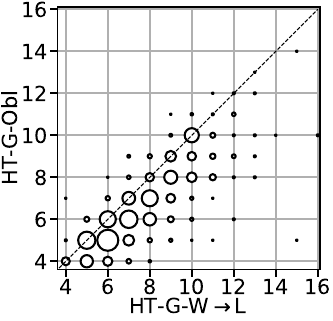}%
 \end{subfigure}%
\caption{Exact and heuristic computations of generalized hypertree decompositions: differences in values depending on the optimization strategy.}
\label{fig:exp_chtw_ccomp}
\end{figure}

Figures~\ref{fig:exp_chtw_ccomp}a to~\ref{fig:exp_chtw_ccomp}c show the results from 259 optimal decompositions computed within the time limit.
The general outlook is the same as for treewidth: Even the \alghtxw\ algorithm %
significantly improves the load
without any trade-off (Figure~\ref{fig:exp_chtw_ccomp}a), and \alghtxl\ can decrease the load even further (Figure~\ref{fig:exp_chtw_ccomp}a) %
while only slightly increasing the generalized hypertree width (Figure~\ref{fig:exp_chtw_ccomp}c).
\looseness=-1
The results obtained by applying the \alghtho\ and \alghthl\ heuristics on the 1624 instances with large width can be seen in Figure~\ref{fig:exp_chtw_ccomp}d. There is a stark contrast to the heuristics used for treewidth: The \alghthw\ heuristic can significantly reduce the load with no trade-off, as the width is guaranteed to be the same (i.e.\ fixed after giving the vertex ordering).
We can lower the load further by optimizing for load first as Figure~\ref{fig:exp_chtw_ccomp}f shows.
Figure~\ref{fig:exp_chtw_ccomp}e shows that the resulting increase in width is about the same as the gain in load.

The results for the greedy heuristic look similar to the branch \& bound results.
Notably, the width is the same for most instances as shown in Figures~\ref{fig:exp_chtw_ccomp}e and h.
The main difference is the slightly increased load as is shown in Figures~\ref{fig:exp_chtw_ccomp}d and \ref{fig:exp_chtw_ccomp}g.
This suggests that the greedy heuristic is a viable choice whenever a slightly higher load is acceptable.

\section{Concluding Remarks}\label{sec:concl}

We have introduced a novel way of refining treewidth and hypertree width via the notion of thresholds, allowing us to lift previous fixed-parameter tractability results for CSP and other problems beyond the reach of classical width parameters. 
Our new parameters
have the advantage over the standard variants of
treewidth and hypertree width that they can take more instance-specific
information into account. A further advantage of our new parameters is
that decompositions that optimize our refined parameter can be used as
the input to existing standard dynamic programming algorithms,
resulting in a potential exponential speedup.
Our empirical findings
show that in realistic scenarios, one can expect that optimizing the
loads requires only minimal overhead while offering huge gains in further processing times.

A natural direction for future research is to explore how the concept of threshold treewidth can be adapted to CSPs in which variables may have infinite domains.
On the one hand, several classes of such CSPs have been shown to be XP-tractable~\cite{huang_decomposition_2013,bodirsky_datalog_2013} and even fixed-parameter tractable~\cite{dabrowski_solving_2021} with respect to the treewidth of the primal graph.
This makes it interesting to attempt to further generalize these tractability results by using the threshold concept.
On the other hand, in the finite-domain regime the potential ``difficulty'' induced by a domain can be captured straightforwardly by its size, however, it seems in the infinite-domain regime the difficulty of a domain has to be captured by different means.
This is indicated when considering Mixed-Integer Linear Programs (MILPs) as CSPs:
Checking the feasibility of MILPs is NP-hard but fixed-parameter tractable with respect to the number of integer variables \cite{lenstra_integer_1983}.
Thus the integer domains introduce the difficulty into checking feasibility rather than the domain size alone.
It thus seems important to capture the structure rather than the size of the domains.
This would need a new approach.

\section*{Acknowledgments}
Andr\'e Schidler and Stefan Szeider acknowledge the support from the
FWF, projects P32441 and W1255,
and from the WWTF, project ICT19-065.
Robert Ganian also acknowledges support from the FWF, notably from projects P31336 and Y1329.
Manuel Sorge acknowledges support by the European Research Council (ERC) under the European Union’s Horizon 2020 research and innovation programme, grant agreement no.~714704 and by the Alexander von Humboldt Foundation.
Main work of Manuel Sorge done while with University of Warsaw.

\begin{center}
  \includegraphics[width=260px]{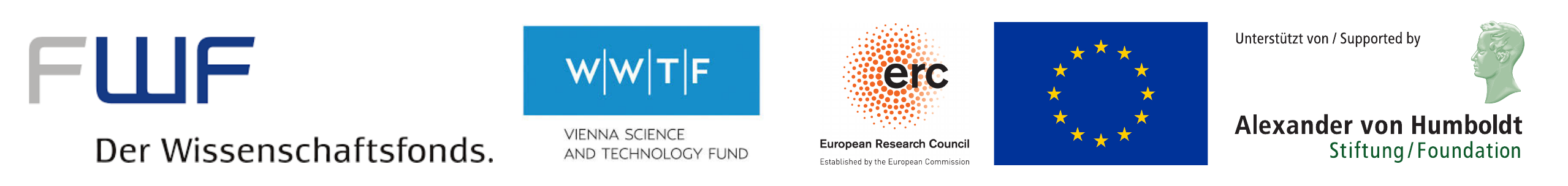}
\end{center}

\bibliography{literature}

\end{document}